\documentclass{llncs}

\usepackage[lncs,amsthm,autoqed]{pamas}
\usepackage{txfonts}
\usepackage{arydshln}
\usepackage{wrapfig}

\newcommand{\set}[1]{\{#1\}}

\newcommand\eat[1]{}

\usepackage{tikz}
\usetikzlibrary{arrows}

\usepackage{mathrsfs}
\newcommand{\W}{$\mathcal{W}$}
\newcommand{\B}{$\mathcal{B}$}
\newcommand{\BB}{B}
\newcommand{\WW}{W}

\newcommand{\nash}{\text{{NS}}\xspace}
\newcommand{\indiv}{\text{{IS }}\xspace}
\newcommand{\contract}{\text{CIS}\xspace}

\usepackage{calrsfs}

\newcommand{\pref}{\succsim\xspace}
\newcommand{\Pref}[1][]{
	\ifthenelse{\equal{#1}{}}{\mathrel \succsim}{\mathop{\succsim_{#1}}}
}                                          
\newcommand{\sPref}[1][]{                  
	\ifthenelse{\equal{#1}{}}{\mathrel \succ}{\mathop{\succ_{#1}}}
}                                          
\newcommand{\Indiff}[1][]{                 
	\ifthenelse{\equal{#1}{}}{\mathrel \sim}{\mathop{\sim_{#1}}}
}
\newcommand{\prefset}[1][]{\ifthenelse{\equal{#1}{}}{\mathcal{\succsim}}{\mathcal{\succsim}_{#1}}}

\newcommand{\midd}{\makebox[2ex]{$|$}}
\newcommand{\middd}{\makebox[2ex]{$\|$}}

\newcommand{\allelse}[1][3em]{\;\rule[.5ex]{#1}{.4pt}\;}

\newcommand{\clauses}{\mathcal X}
\newcommand{\vars}{\mathcal P}

\newcommand{\form}{\varphi}
\newcommand{\clause}{X}
\newcommand{\prop}{p}
\newcommand{\lit}{x}
\newcommand{\valuation}{v}

\newlength{\wordlength}

\newcommand{\SAT}{\text{{\sc Sat}}\xspace}

\title{Individual-based stability in hedonic games depending on the best or worst players
\thanks{This material is based upon work supported by the Deutsche Forschungsgemeinschaft under grants BR-2312/6-1 (within the European Science Foundation's EUROCORES program LogICCC) and BR~2312/7-1. Many thanks to Felix Brandt and Matthew Jackson for some useful discussions on coalition formation.
}}
\author{Haris Aziz \and Paul Harrenstein \and Evangelia Pyrga}
\institute{%
  Instit\"ut f\"ur Informatik,
  Technische Universit\"at M\"unchen, 
  85748 M\"unchen, Germany \\
  \email{\{aziz,harrenst,pyrga\}@in.tum.de}}
\begin{document}

\maketitle



\begin{abstract}
	
	
	We consider hedonic coalition formation games in which each player has preferences over the other players and his preferences over coalition structures are based on the best player ($\mathcal{B}$-/B-hedonic games) or the worst player ($\mathcal{W}$/W-hedonic games) in his coalition.
		We show that for \B-hedonic games, an individually stable partition is guaranteed to exist and can be computed efficiently. Similarly, there exists a polynomial-time algorithm which returns a Nash stable partition (if one exists) for \B-hedonic games with strict preferences. 
	It is also shown that for \BB- and \W-hedonic games, 
	checking whether a Nash stable partition or an individually stable partition exists is NP-complete even in some cases for strict preferences. 
	As a result of our investigation, we identify a key source of intractability in compact coalition formation games in which preferences over players are extended to preferences over coalitions.
\end{abstract}

%
%
%

\section{Introduction}

%
%
%
%
%
%
%
%
%
%
%
%
%
%
%
%
%
%
%
%
%
%
%
%
%

\begin{table*}
       \centering
       \scalebox{1.1}{
               \begin{tabular}{lllll}
                       \toprule
                       && \nash & \indiv&       \contract \& IR\\\midrule
                  \B&(general)                                            &?&in P  (Th.~\ref{th:is-exists})                                 &    in P (Prop.~\ref{th:W-CIS})               \\
                       \B&(strict preferences)                         &in P (Th.~\ref{th-ns-strict-b})                                &           in P  (Th.~\ref{th:is-exists})                                                                      & in P (Prop.~\ref{th:W-CIS})    \\
\midrule
                       \BB&(general)                                            &NPC (Th.~\ref{theorem:ww-bb-srict-ns-hard})                                                                                             & NPC (Th.~\ref{theorem:b-is-hard})                                                                               &in P (Prop.~\ref{th:W-CIS})\\
                      \BB&(strict preferences)                         &NPC (Th.~\ref{theorem:ww-bb-srict-ns-hard})    &NPC (Th.~\ref{theorem:b-is-hard})              &in P (Prop.~\ref{th:W-CIS}) \\
                       \BB&(no unacceptability)                         &in P (Obs.~\ref{obs:GC-Nash})                                  &in P (Obs.~\ref{obs:GC-Nash})                  &in P (Prop.~\ref{th:W-CIS}) \\\midrule

              \W/\WW &(general)                                            &NPC (Th.~\ref{theorem:ww-bb-srict-ns-hard})                              &NPC (Th.~\ref{theorem:w-is-hard})              &in P (Prop.~\ref{th:W-CIS})\\
                       \W/\WW &(strict preferences)                         &NPC (Th.~\ref{theorem:ww-bb-srict-ns-hard})    &     ?                                                                          &in P (Prop.~\ref{th:W-CIS}) \\
                       \W/\WW &(no unacceptability)                         &in P (Obs.~\ref{obs:GC-Nash})                                  &in P (Obs.~\ref{obs:GC-Nash})                  &in P (Prop.~\ref{th:W-CIS}) \\\midrule

                       
               \end{tabular}
       }
\caption{Complexity of individual-based stability: the positive results even hold for computation of stable partitions whereas the NP-completeness results even hold for checking the existence of a stable partition.}
\label{table:bestworst-complexity}
\end{table*}

Coalition formation plays a fundamental role in various multiagent settings.
The following quotation from \citep{BoJa02a} nicely highlights the significance of coalition formation:
\begin{quote} \emph{``Coalition formation is of fundamental importance in a wide variety of
social, economic, and political problems, ranging from communication and
trade to legislative voting. As such, there is much about the formation
of coalitions that deserves study.''}\end{quote}

\emph{Coalition formation games}, as introduced by \citet{DrGr80a}, provide a simple but versatile formal model that allows one to focus on coalition formation. In many situations it is natural to assume that a player's appreciation of a coalition structure only depends on the coalition he is a member of and not on how the remaining players are grouped. 
Initiated by  \citet{BKS01a} and \citet{BoJa02a}, much of the work on coalition formation now concentrates on these so-called \emph{hedonic games}. 

The main focus in hedonic games has been on notions of stability for coalition structures. \citet{BoJa02a} formalized individual-based stability concepts (\emph{Nash stability (NS)}, \emph{individual stability (IS)}, and \emph{contractual individual stability (CIS)}) and group-based stability concepts (\emph{core (C)} and \emph{strict core (SC)}) in the context of hedonic games. 
The most prominent examples of hedonic games are two-sided matching games in which only coalitions of size two are admissible~\citep{RoSo90a}. 
We refer to~\citet{Hajd06a} and \citet{Cech08a} for a critical overview of hedonic games. 

Hedonic games involve preferences over sets of players. This leads to the challenge of succinctly representing hedonic games~\citep[see e.g., ][]{ElWo09a}. 
A natural method to represent them is for each player to have preferences over the individual players and then extend these preferences to preferences over sets of players. 
This method includes \emph{additively separable} preferences in which cardinal utilities are considered~\citep{ABS11b,BoJa02a}. 
Other natural ways to represent preferences over coalitions is to base them on the best player in the coalition (e.g., \B-preferences) or on the worst player in the coalition (e.g., \W-preferences). Preferences over sets of objects based on the worst or best object in the set is a standard method for preference extension which satisfies various desirable axioms~\citep{BBP04a}. Such preference extensions can serve as building blocks to form more complex preference extensions.


\B-preferences and \W-preferences are two of the most natural ordinal-based compact representations for hedonic games~\citep{CeRo01a,CeHa02a,CeHa04a}. 
\B-preferences model scenarios where players are `optimistic' in nature and each player's evaluation of a coalition depends on his most favored players in that coalition~\citep{Hajd06a}. 
On the other hand, \W-preferences model `pessimistic' players where the player's happiness depends on the least favored players in his coalition. This is based on the premise that \emph{`a team is only as good as the weakest link'.}
We refer to hedonic games with \B-preferences and \W-preferences as \B-hedonic games and \W-hedonic games respectively.

In many distributed settings, agents have limited communication with other agents or insufficient trust over others which motivates the use of individual-based stability concepts rather than group-based stability concepts to model coalition formation. Compared to the group-based stability concepts such as the core, the complexity of individual-based stability concepts has not been examined for \W- and \B-hedonic games. In fact \citet{Hajd06a} points out that 
\begin{quote}\emph{``It would be interesting to study the computational complexity questions also
for problems of deciding the existence of a Nash stable, an individually stable
and a contractually individually stable partition in games with \B, \W, $[\ldots]$ preferences.''} \end{quote}
We answer the questions of \citet{Hajd06a} and present a thorough investigation of individual-based stability in hedonic games in which preferences are derived from the best or worst player in the coalition. 
Our analysis includes positive results like the guaranteed existence of individually stable partitions for \B-hedonic games and also an insight into a key cause of intractability of stable partitions in hedonic games based on preferences over players.

Along with \B- and \W-hedonic games, we consider two new variants, namely
the \BB- and \WW-hedonic games, 
with the main difference being how unacceptability of
a player is perceived. In particular, while the presence of an unacceptable
player in a coalition in \B-hedonic games becomes irrelevant under the presence
of other players in the same coalition, in \BB- and \WW- hedonic games a single 
unacceptable player suffices to make the coalition unattractive.
\WW-hedonic games coincide with \W-hedonic games when individual
rational outcomes are considered. On the other hand, \BB-hedonic games are not
equivalent to \B-hedonic games: 
unacceptability of coalitions is defined differently and the sizes of
coalitions are not important.
\BB-hedonic games present a number of interesting properties. They are defined
in a symmetric way to \W- for which the coalition size does not affect the
preferences of players. Secondly, marriage games are a subclass of the
intersection of \BB- and \W-hedonic games and this highlights the symmetry
between \BB- and \WW-hedonic games. Thirdly, our definition of \BB-hedonic games depends very
naturally on the meaning of unacceptable players, i.e., if a player $i$ finds
player $j$ unacceptable, he would rather be alone than ever be with $j$ even in
the presence of other players. 
Interestingly, our computational analysis of \BB-hedonic games and \B-hedonic games helps highlight
the source of intractability in a large family of coalition formation games. 

\paragraph{Contributions}

We present a number of computational results regarding individual-based stability in very natural models of coalition formation games.
We show that for \B-hedonic games, at least one individually stable partition exists and it can be computed in linear time. A Nash stable partition may not exist. However, it can be checked in linear time whether a Nash stable exists for \B-hedonic games with strict preferences, 
and in case of existence, such a partition can be returned. It is also shown that a contractually individually stable and individually rational partition can be computed in polynomial time for all the games considered.

For \BB- and \W-hedonic games, 
checking whether a Nash stable partition exists is NP-complete even when preferences are strict. 
Also, for \BB- and \W-hedonic games, 
checking whether an individually stable partition exists is NP-complete. For \BB-hedonic games, the result holds even if preferences are strict.

We obtain a general insight that \emph{in coalition formation games based on extensions of preferences over players to preferences over sets, the following property leads to intractability:
the presence of an unacceptable player makes the coalition unacceptable}. Our results are summarized in Table~\ref{table:bestworst-complexity}, which gives an almost complete
characterization of the complexity of individual-based stability in \B-, \BB- and \W-hedonic games. Since we consider very restrictive models of coalition formation games, our hardness results carry over to more elaborate settings.

\section{Related work}

Computational complexity of computing stable outcomes for coalition formation games is an active area of research in theoretical computer science~\citep{Cech08a}, game theory~\citep{Hajd06a}, and multiagent systems~\citep[see \eg][]{ABS11b, Gasa11a}. Coalition formation games have also received attention from the artificial intelligence community, where the focus has generally been on computing partitions that give rise to the greatest social welfare \citep[\eg][]{SLA+99a}. 

Both \W- and \B-hedonic games were first introduced by \citet{CeRo01a}. 
Since then, the complexity of computing and verifying stable partitions for \W- and \B-hedonic games has been studied~\citep[see \eg][]{CeHa02a,CeHa04a}. For \W-hedonic games, it is NP-hard to check whether the core is non-empty. The core can be empty even if preferences are strict~\citep{CeHa04a}. 
For \B-hedonic games, \citet{CeHa02a} showed that it is NP-hard to check whether the core or strict core is empty. 
For \B-hedonic games with strict preferences, \citet{CeRo01a} proved that the strict core is non-empty and a strict core stable partition can be computed in polynomial time by a generalization of Gale's top trading cycle algorithm~\cite{ShSc74a}.

Recently, \citet{ABH11b} examined the complexity of Pareto optimal partitions for \BB- and \WW-hedonic games. They showed that although computing a Pareto optimal partition is NP-hard for \BB-hedonic games, there exists a polynomial-time algorithm to solve the problem for \WW-hedonic games.

The complexity of individual-based stability has been investigated previously for hedonic games represented by individually rational lists of coalitions~\citep{Ball04a} and also additively separable hedonic games~\citep[see \eg ][]{Olse09a,Gasa10a,SuDi10a, Gasa11a}. In particular, in a paper in last year's AAMAS, \citet{Gasa11a} examined individual-based stability in symmetric additively separable hedonic games.
However, none of the results for other representations of hedonic games imply any of the results for hedonic games based on the best or worst players. Our focus on hedonic games based on the best or worst player is motivated
by the natural ordinal nature of the games, their succinct representation, 
and the hardness of deciding the existence of or computing stable outcomes for other succinct representation of hedonic games.
\section{Preliminaries}

In this section, we review the terminology and notation used in this paper.

\paragraph{Hedonic games}
Let~$N$ be a set of~$n$ players. A \emph{coalition} is any non-empty subset of~$N$. By $\mathcal N_i$ we denote the set of all coalitions player~$i$ may belong to, \ie $\mathcal N_i=\set{S\subseteq N\midd i\in S}$. 
A \emph{coalition structure}, or simply a \emph{partition}, is a partition~$\pi$ of the players $N$ into coalitions, where $\pi(i)$ is the coalition player~$i$ belongs to.

A \emph{hedonic game} is a pair~$(N,\pref)$, where $\pref=(\pref_1,\dots,\pref_n)$ is a \emph{preference profile} specifying the preferences of each player~$i$ as a binary, complete, reflexive, and transitive \emph{preference relation~$\Pref[i]$} over~$\mathcal N_i$. If~$\pref_i$ is also anti-symmetric we say that~$i$'s preferences are \emph{strict}. Note that $S\sPref[i]T$ if $S\Pref[i]T$ but not~$T\Pref[i]S$---\ie if~$i$ \emph{strictly prefers} $S$ to~$T$---and $S\Indiff[i]T$ if both $S\Pref[i]T$ and $T\Pref[i]S$---\ie if~$i$ is \emph{indifferent} between~$S$ and~$T$.

For a player~$i$, a coalition~$S$ in~$\mathcal N_i$ is \emph{acceptable} if for~$i$ being in~$S$ is at least as preferable as being alone---\ie if $S\Pref[i]{\set i}$---and \emph{unacceptable} otherwise. If in player $i$'s preference over other players, $j \sPref_i i$, then we say that $i$ \emph{likes} $j$.

In a similar fashion, for~$X$ a subset of~$\mathcal N_i$, a coalition~$S$ in~$X$ is said to be \emph{most preferred in~$X$ by~$i$} if $S\Pref[i]T$ for all $T$ in~$X$ and \emph{least preferred in~$X$ by~$i$} if $T\Pref[i]S$ for all $T\in X$. In case~$X=\mathcal N_i$ we generally omit the reference to~$X$. The sets of most and least preferred coalitions in~$X$ by~$i$, we denote by $\max_{\Pref[i]}(X)$ and $\min_{\Pref[i]}(X)$, respectively.  

In hedonic games, players are only interested in the coalition they are in. Accordingly, preferences over coalitions naturally extend to preferences over partitions and we write $\pi\Pref[i]\pi'$ if $\pi(i)\Pref[i]\pi'(i)$.
We also say that partition~$\pi$ is \emph{acceptable} or \emph{unacceptable} to a player~$i$ according to whether~$\pi(i)$ is acceptable or unacceptable to~$i$, respectively. Moreover,~$\pi$ is \emph{individually rational (IR)} if~$\pi$ is acceptable to all players.

\paragraph{Classes of hedonic games}
The number of coalitions grows exponentially in the number of players. 
In this sense, hedonic games are relatively large objects and for algorithmic purposes it is often useful to look at classes of games that allow for concise representations. 

We now describe classes of hedonic games in which the players' preferences over coalitions are naturally induced by their preferences over the other players. Therefore, we will use the same notation $\pref_i$ for each player $i$'s preferences over players and also over coalitions.
For~$\pref_i$, preferences of player~$i$ over players, we say that a player~$j$ is \emph{acceptable} to~$i$ if $j\mathop{\pref_i}i$ and \emph{unacceptable} otherwise. 



For a subset~$J$ of players, we denote by $\max_{\pref_i}(J)$ and $\min_{\pref_i}(J)$ the sets of the most and least preferred players in~$J$ by~$i$, respectively. We will assume that $\max_{\pref_i}(\emptyset)=\min_{\pref_i}(\emptyset)=\{i\}$. Let $i\in N$ and let $S, T \in {\mathcal N_i}$.

\begin{itemize}
\item In a \emph{\BB-hedonic game}, the preferences~$\pref_i$ of a player~$i$ over players extend to preferences over coalitions in such a way that 
we have $S\Pref[i]T$ if and only if 
\begin{enumerate}
\item some $j$ in~$T$ is unacceptable to~$i$ or 
\item neither $S$ nor $T$ contains a player unacceptable to $i$ and for each $s\in \max_{\pref_i}(S\setminus \{i\})$ and $t\in \max_{\pref_i}(T\setminus \{i\})$,  $s\succsim_i t$.
\end{enumerate}



\item In a \emph{\WW-hedonic game $(N,\pref)$}, we have $S\Pref[i]T$ if and only if 
\begin{enumerate}
	\item some $j$ in~$T$ is unacceptable to~$i$, or
	\item for each $s\in \min_{\pref_i}(S\setminus \{i\})$ and $t\in \min_{\pref_i}(T\setminus \{i\})$,  $s\succsim_i t$.
	
\end{enumerate}

\item In \emph{hedonic games with \W-preferences} (which we will refer to as \W-hedonic games), $S\succsim_i T$ if and only if

\begin{itemize}
\item[]	 for each $s\in \min_{\pref_i}(S\setminus \{i\})$ and $t\in \min_{\pref_i}(T\setminus \{i\})$,  $s\succsim_i t$. 
\end{itemize}

%
\item In \emph{hedonic games with \B-preferences} (which we will refer to as \B-hedonic games), $S \sPref[i] T$ if and only if 
\begin{enumerate}
	\item for each $s\in \max_{\pref_i}(S\setminus \{i\})$ and $t\in \max_{\pref_i}(T\setminus \{i\})$, $s\succ_i t$, or
	\item for each $s\in \max_{\pref_i}(S\setminus \{i\})$ and $t\in \max_{\pref_i}(T\setminus \{i\})$, $s\sim_i t$ and $|S|<|T|$.
\end{enumerate}
\end{itemize}

\WW-hedonic games are equivalent to hedonic games with $\mathcal W$\!-preferences if only individually rational outcomes are considered. Unlike hedonic games with $\mathcal B$-preferences, \BB-hedonic games are defined in analogy to \WW-hedonic games and the preferences are not based on coalition sizes~\citep[cf. ][]{CeRo01a}.

\begin{example}
Consider the preferences of player $1$ over other players:
$$2\succ_1 3 \succ_1 1 \succ_1 4$$
Then, in corresponding \BB-, \WW-, \B-, and \W-hedonic game, the preferences of player $1$ over coalitions which include $1$ are induced as follows.
\begin{itemize}
\item \BB-hedonic game: $\{1,2\}\sim_1 \{1,2,3\}\succ_1 \{1,3\}\succ_1 \{1\} \succ_1 \{1,4\}\sim_1 \{1,2,4\}\sim_1 \{1,3,4\}\sim_1 \{1,2,3,4\}$.
\item \WW-hedonic game: $\{1,2\}\succ_1 \{1,2,3\}\sim_1 \{1,3\}\succ_1 \{1\} \succ_1 \{1,4\}\sim_1 \{1,2,4\}\sim_1 \{1,3,4\}\sim_1 \{1,2,3,4\}$.
\item \B-hedonic game: $\{1,2\}\succ_1 \{1,2,3\}\sim_1 \{1,2,4\}\succ_1 \{1,2,3,4\} 
\succ_1 \{1,3\}\succ_1 \{1,3,4\}\succ_1 \{1\} \succ_1 \{1,4\}$.
\item \W-hedonic game: $\{1,2\}\succ_1 \{1,2,3\}\sim_1 \{1,3\}\succ_1 \{1\} \succ_1 \{1,4\}\sim_1 \{1,2,4\}\sim_1 \{1,3,4\}\sim_1 \{1,2,3,4\}$.
\end{itemize}

Note that player $1$'s preferences over other players lead to the same preferences over coalitions in both \WW- and \W- hedonic game. 
\end{example}




\paragraph{Stability Concepts}

We now present the standard stability concepts for hedonic games. 
A partition $\pi$ is \emph{individually rational (IR)} if each player does as well in his current coalition as by being alone, i.e., for all $i\in N$, $\pi(i) \succsim_i  \{i\}$. 
The following are standard stability concepts based on deviations by individual players.
	\begin{itemize}
\item A partition is  \emph{Nash stable (NS)} if no player can benefit by 
	moving from his coalition $S$ to another (possibly empty) coalition $T$.
\item A partition is \emph{individually stable (IS)} if no player can
	benefit by moving from his coalition $S$ to another existing (possibly empty) coalition $T$  while not making the members of $T$ worse off. 
	\item A partition is  \emph{contractually individually stable (CIS)} if no
	player can benefit by moving from his coalition $S$ to another existing (possibly empty) coalition
	$T$ while making neither the members of $S$ nor the members of
	$T$ worse off.
	\end{itemize}
	
	\begin{figure}[b]
    \centering
    \scalebox{1.2}{
	\begin{tikzpicture}
		\tikzstyle{pfeil}=[->,>=angle 60, shorten >=1pt,draw]
		\tikzstyle{onlytext}=[]
		
		\node[onlytext] (NS) at (1,0) {NS};
		\node[onlytext] (SC) at (3,0) {SC};
		\node[onlytext] (IS) at (2,-1.5) {IS};
		\node[onlytext] (C) at (4,-1.5) {C};
		\node[onlytext] (CIS) at (2,-3) {CIS \& IR};
	
		\draw[pfeil] (NS) to (IS);
		\draw[pfeil] (SC) to (IS);
		\draw[pfeil] (IS) to (CIS);
		\draw[pfeil] (SC) to (C);
	\end{tikzpicture}
	}
    \caption{Inclusion relationships between stability concepts. For example, every Nash stable partition is also individually stable.}\label{fig:tnfigure}
\end{figure}
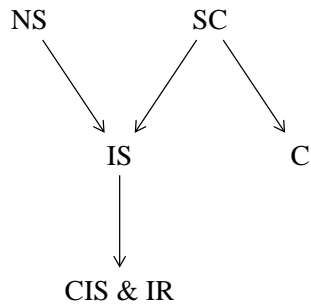
	
	We also define standard stability concepts based on deviations by \emph{groups} of players.
\begin{itemize}
\item A coalition $S \subseteq N$ \emph{blocks} a partition $\pi$, if each
player $i \in S$ strictly prefers $S$ to his current coalition $\pi(i)$ in
the partition $\pi$. A partition which admits no blocking coalition is said to be in the \emph{core (C)}.
\item A coalition $S \subseteq N$ \emph{weakly blocks} a partition $\pi$,
if each player $i \in S$ weakly prefers $S$ to $\pi(i)$ and there exists at least one player $j \in S$ who strictly prefers $S$ to his current coalition $\pi(j)$. A partition which admits no weakly blocking coalition is in the \emph{strict core (SC)}. 
\end{itemize}
The inclusion relationships between stability concepts depicted in Figure~1 follow from the definitions of the concepts. Note that depending on the context, we will denote by NS, IS, and CIS either Nash stable, individually stable, and contractual individually stable respectively or Nash stability, individual stability, and contractual individual stability respectively.


\eat{
\begin{figure}
\begin{center}
	\scalebox{0.5}{
	\begin{tikzpicture}
	\Large
		\tikzstyle{pfeil}=[->,>=angle 60, shorten >=1pt,draw]
		\tikzstyle{onlytext}=[]
		
		\node[onlytext] (NS) at (1,0) {NS};
		\node[onlytext] (SC) at (3,0) {SC};
		\node[onlytext] (IS) at (2,-1.5) {IS};
		\node[onlytext] (C) at (4,-1.5) {C};
		\node[onlytext] (CIS) at (2,-3) {CIS \& IR};
	
		\draw[pfeil] (NS) to (IS);
		\draw[pfeil] (SC) to (IS);
		\draw[pfeil] (IS) to (CIS);
		\draw[pfeil] (SC) to (C);
	\end{tikzpicture}
	}
	\end{center}

	\caption{Inclusion relationships between stability concepts. For example, every Nash stable partition is also individually stable.}
	\label{fig:relations}
	\end{figure}	
}
\normalsize

	We first note that a partition which is both CIS and IR can be computed in polynomial time for both \BB- and \W-hedonic games:

\begin{proposition}\label{th:W-CIS}
A CIS and IR partition can be computed in polynomial time for \B-,\BB-, \W-, and \WW-hedonic games. 
\end{proposition}
\begin{proof}
Take the individually rational partition of singletons. 
If the partition is CIS, we are done. 
Otherwise, if there is a feasible CIS deviation, we let the deviation take place. 
In each CIS deviation at least one player strictly improves his utility and no player's utility decreases. 
Since there can only be a maximum of a polynomial number of CIS deviations ($n(n-1)$ for \BB-, \W-, and \WW-hedonic games---and $n^2(n-1)$ in the case of \B-hedonic games), a CIS and IR partition is obtained in polynomial time. 
\end{proof}

Moreover, for \B-, \BB-, \W- and \WW-hedonic games, individual-based stability can be verified in polynomial-time. This means that checking the existence of individual-based stability is in NP. Since, \WW-hedonic games coincide with \W-hedonic games when individually rational outcomes are considered, all our results for Nash stability and individual stability equivalently apply to \WW- and \W-hedonic games. Therefore, from now on, we will only focus on \W-hedonic games rather than both \WW- and \W-hedonic games.

\section{\BB- and \W-hedonic games: Nash stability}

In this section, we consider Nash stability in \BB- and \W-hedonic games.

\begin{observation}\label{obs:GC-Nash}
For both \W- and \BB-hedonic games, if preferences contain no unacceptable players, then the partition consisting of the grand coalition is Nash stable and therefore individually stable.
\end{observation}

We present an example of a \W- or \BB-hedonic game for which there is no NS partition. This holds even if preferences are strict but do allow unacceptability.
\begin{example}
Consider the hedonic game $(N,\pref)$ where $N=\{1,2\}$ such that $1$ has no other acceptable players and player $2$ likes the company of $1$. Therefore, 
$\{1\}\succ_1 \{1,2\}$ and $\{1,2\}\succ_2 \{2\}$. 	We see that the game has no NS partition since player $2$ joins player $1$ and then player $1$ leaves player $2$ alone. We call this game the \emph{stalker game} where player $1$ is clearly being stalked.
\end{example}

%

We first prove that for both \W- and \BB-hedonic games, deciding whether a NS partition exists is NP-complete. 
Theorem~\ref{theorem:B-FP-hard} will be subsumed by a later result which states that for \W-hedonic and \BB-hedonic games, deciding whether a NS partition exists is NP-complete even if preferences are strict. However, we present Theorem~\ref{theorem:B-FP-hard} to provide better intuition and as a warm-up for the technically more involved proofs of Theorems~\ref{theorem:b-is-hard} and \ref{theorem:ww-bb-srict-ns-hard}.

\begin{theorem}\label{theorem:B-FP-hard}
	For \W-hedonic and \BB-hedonic games, deciding whether a NS partition exists is NP-complete. 
\end{theorem}

\begin{proof}
	By a reduction from \SAT~\citep{GaJo79a}. Let $\form=X_1\wedge\dots\wedge\clause_k$ be a Boolean formula in conjunctive normal form in which all and only the Boolean variables~$\prop_1,\dots,\prop_m$ occur. Now define the \BB-hedonic game $(N,\pref)$, where 
$N=\set{\clause_1,\dots,\clause_k}\cup\set{\prop_1,\neg\prop_1,\dots,\prop_m,\neg\prop_m}\cup\set{0,1}.$
	
	The main idea is to design the ``clause'' players ($\clause_1, \ldots, \clause_k$) so as to be stalkers of player~$1$, like in the stalker game.
Define the preferences~$\pref$ such that for each literal~$\prop$ or $\neg\prop$, and each clause $\clause=(\lit_1\vee\dots\vee\lit_\ell)$,
\[
\renewcommand{\arraystretch}{1.2}
	\begin{array}[b]{r@{\quad}l}
		\prop\colon& (0,1,\;\rule[.5ex]{3em}{.4pt}\;,\phantom{\neg}\prop\middd\neg\prop,\clause_1,\dots,\clause_k)	\\
	\neg\prop\colon& (0,1,\;\rule[.5ex]{3em}{.4pt}\;,\neg\prop\middd\phantom\neg\prop,\clause_1,\dots,\clause_k)	\\							
		\clause\colon& (1,\;\rule[.5ex]{3em}{.4pt}\;\midd \clause_1,\dots,\clause_k\middd 0, \lit_1,\dots,\lit_\ell)	\\
			0\colon& (\;\rule[.5ex]{3em}{.4pt}\;,0\middd 1,\clause_1,\dots,\clause_k)	\\				
			1\colon& (\;\rule[.5ex]{3em}{.4pt}\;,1\middd 0,\clause_1,\dots,\clause_k)\text,
	\end{array}
\]
where the horizontal lines stand for the players not explicitly mentioned in the list, the vertical lines divide the players into equivalence classes, and the players after the double vertical lines are unacceptable. We have introduced the new notation to improve readability: in the preference list of player $i$, vertical bars are for $\succ_i$ and commas are for $\sim_i$. 
 
We prove that $\form$ is satisfiable if and only if a NS (and IR) partition for $(N,\pref)$ exists. 

To this end, first assume that~$\valuation$ is a valuation that satisfies~$\form$. Then, define the partition~$\pi$ such that
$$\pi=\{\{1,\lit'_1,\dots,\lit'_{\ell'}\}, \{0,\lit''_1,\dots,\lit''_{\ell''}\}, \{\clause_1,\dots,\clause_k\}\}$$

where $\lit'_1,\dots,\lit'_{\ell'}$ are the literals rendered true by~$\valuation$ and $\lit''_1,\dots,\lit''_{\ell''}$ those that are thus rendered false. Obviously, $\pi$ is a favorite partition for both~$0$ and~$1$ and as such they do not want to deviate. With~$\valuation$ being well-defined as a valuation, no two ``literal'' players~$\prop$ and $\neg\prop$ are in the same coalition, \ie $\prop\in\pi(0)$ if and only if $\neg\prop\in\pi(1)$. Thus, every ``literal'' player is in a favorite coalition and does not want to deviate either. For each 
``clause'' player~$\clause=(\lit_1\vee\dots\vee\lit_\ell)$, both $\pi(0)$ and~$\pi(1)$ are unacceptable---the former, because $0\in\pi(0)$, the latter, because~$\valuation$ satisfies~$F$ and, therefore, at least one of~$\lit_1,\dots,\lit_\ell$ is in~$\pi(1)$. Thus, no player wishes to deviate to another coalition.

For the opposite direction, first observe that, if there is an individually rational partition~$\pi$ such that for each clause~$\clause=(\lit_1\vee\dots\vee\lit_\ell)$  at least one of~$\lit_1,\dots,\lit_\ell$ is  in~$\pi(1)$, the valuation~$\valuation$ that sets all literals in~$\pi(1)$ to true, satisfies~$\form$. In particular, observe that~$\valuation$ is thus well-defined, as for no Boolean variable~$\prop$, both $\prop$ and $\neg\prop$ can both be in~$\pi(1)$ without violating individual rationality.

Now assume that there is no valuation satisfying~$\form$ and consider an arbitrary NS partition structure~$\pi$. $\pi$ must also be IR. Then, by the previous observation, for some ``clause'' player $\clause=(\lit_1\vee\dots\vee\lit_\ell)$ none of~$\lit_1,\dots,\lit_\ell$ is in~$\pi(1)$, nor is~$0$,~$\clause$ himself, or any other ``clause'' player, as each of these are unacceptable to~$1$. Also by individual rationality, none of the ``literal'' players~$\prop$ nor~$0$ is in~$\pi(\clause)$. It follows that~$\clause$ would like to deviate and join~$\pi(1)$, so as to improve both the best and the worst player in his coalition. In that case though, we have an instance of the stalker game, with $X$ being the stalker of player $1$. It follows that~$\pi$ is not NS.
\end{proof}


\section{\BB- and \W-hedonic games: individual stability}

We first show that \BB- and \W-hedonic games may not admit an IS partition even if preferences are strict:

\begin{example}\label{example:B-IS-strict-counter}
Consider the game $(N,\pref)$ such that $N=\{1,\ldots, 5\}$ and $\pref$ is strict but allows unacceptability. For each player $i\in \{2,3,4\}$, his preferences look as following $i: i+1 \sPref_i i-1 \sPref_i i\succ_i \cdots$. Similarly, $1 \sPref_5 4 \sPref_5  5 \succ_5 \cdots$, and $2 \sPref_1 5 \sPref_1 1 \succ_1 \cdots$. 
We know that there are no IR coalitions of size more than or equal to $3$. The partition of singletons is certainly not IS. In fact, any IR partition and also potentially IS partition consists of one singleton and two acceptable pairs. Then, it can be seen that any IR and potentially IS partition cycles via IS deviations. Without a loss of generality, assume that the starting partition is $\{\{2,3\}, \{4,5\},\{1\}\}$. Then, the IS deviations lead to the following series of partitions:

\begin{enumerate}
\item $\{\{1\}, \{2,3\},\{4,5\}\}$; 
\item $\{\{5,1\}, \{2,3\},\{4\}\}$; 
\item $\{\{5,1\}, \{2\},\{3,4\}\}$; 
\item $\{\{1,2\}, \{3,4\},\{5\}\}$; 
\item $\{\{1,2\}, \{3\},\{4,5\}\}$; and then again 
\item $\{\{1\}, \{2,3\},\{4,5\}\}$.
\end{enumerate}
We will call $(N,\pref)$ the \em{extended stalker game}.
\end{example}


From Example~\ref{example:B-IS-strict-counter}, we know that if preferences are strict but allow unacceptable players, there may not exist any IS or strict core stable partition.


\begin{theorem}\label{theorem:b-is-hard}
Checking whether an IS partition exists is NP-complete for \BB-hedonic games with
strict preferences.
\end{theorem}
\begin{proof}

By a reduction from \SAT. Let $\form=X_1\wedge\dots\wedge\clause_k$ be a Boolean
formula in conjunctive normal form in which all and only the Boolean
variables~$\prop_1,\dots,\prop_m$ occur.  Let $\clauses = \{X_1,\ldots,X_k\}$ 
and $\vars= \{p_1,\ldots,p_m\}$.  Without loss of generality, we will assume that there is 
no  clause $X\in \clauses$ and variable $p$, such that $X$ contains both literals
$p$ and $\neg p$. Also, we assume that for each variable $p$, both  literals 
$p, \neg p$ appear in $\form$.  For any clause $X\in\clauses$, 
let $L^X_+, L^X_-$ denote the sets of all positive and all negative literals in
$X$, respectively and let $L^X = L^X_+ \cup L^X_-$.  Now, we define the \BB-hedonic game $(N,\pref)$, where 


\begin{align*}
N & =  \{1^X, 2^X, 3^X, 4^X,5^X\mid X\in \clauses\} \cup 
	 \bigcup_{p\in\vars}\{ 0_p \}			\\
	&\cup \bigcup_{X\in \clauses}\{ p^X \midd p \in L^X_+ \}	
	\cup  \bigcup_{X\in \clauses}\{ \neg p^X | \neg p \in L^X_- \}
	 .\end{align*}

The main idea will be that unless $\form$ is satisfiable, an instance of the 
extended  stalker game will appear.
	
Let $C_p = \bigcup_{X\in \clauses}\{ p^X \midd p \in L^X_+ \}$, 
$C_{\neg p} = \bigcup_{X\in \clauses}\{\neg p^X |\neg p \in L^X_- \}$. 
$C_p$,  ($C_{\neg p}$) is the  set of players that correspond to all copies of
the positive (negative) literal $p$ (one copy for each clause in which the
literal appears). Also, let  $P =\bigcup_{p\in \vars}(C_p\cup C_{\neg p})$.
	
	Define the preferences~$\pref$ such that for each variable ~$\prop$, and each clause $X\in\clauses$,

\begin{align*}
p^X 	&: (0_p \mid  C_p\setminus\{p^X\} \mid 1^X \mid  
		P\setminus(C_p\cup C_{\neg p}) 	  \mid p^X  \middd \allelse) 	\\
\neg p^X &: (0_p \mid  C_{\neg p}\setminus\{\neg p^X\} \mid 1^X \mid  
		P\setminus(C_p\cup C_{\neg p}) \mid \neg p^X   \middd \allelse) 	\\
0_p &: (C_p \cup C_{\neg p}, \mid  0_p  \middd \allelse) 
 	\\
1^X 	&:  (2^X \mid  L^X \mid 5^X \mid  1^X  \middd \allelse) 	\\
2^X 	&:  (3^X \mid 1^X \mid  2^X  \middd \allelse) 	\\
3^X 	&:  (4^X \mid 2^X \mid  3^X  \middd \allelse) 	\\
4^X 	&:  (5^X \mid 3^X \mid  4^X  \middd \allelse) 	\\
5^X 	&:  (1^X \mid 4^X \mid  5^X  \middd \allelse) 	\\
\end{align*}

\noindent
where in the above, using a set of players $S$ in the preference of a
player implies any arbitrary ordering of the players in $S$. In the preference lists, like in the proof of 
Theorem~\ref{theorem:B-FP-hard}, the vertical lines divide the players into equivalence classes, and the players after the double vertical lines are unacceptable. Recall that we have introduced the new notation to improve readability: in the preference list of player $i$, vertical bars are for $\succ_i$ and commas are for $\sim_i$. Also, the horizontal lines stand for any arbitrary ordering of the players not explicitly
mentioned in the list.

We show that $\form$ is satisfiable if and only if $(N,\pref)$ has a non-empty set of IS partitions.

First assume that $\form$ is satisfiable and $\valuation$ is a valuation that
satisfies~$\form$. Let 
$T_\valuation=  \bigcup_{p:\valuation(p)=\textrm{true}} C_p  $,
$F_\valuation=  \bigcup_{p:\valuation(p)=\textrm{false}} C_{\neg p}  $.
Then, define the partition~$\pi$ as

\begin{align*}
          \{\{2^X,3^X\} \mid X\in \clauses  \}\cup \{\{4^X,5^X\} & \mid X\in \clauses \} \cup\\
          \{\{1^X\} \cup (T_\valuation \cap L^X)  & \mid X\in \clauses\} \cup\\
          \{\{0_p\}\cup (F_\valuation \cap(C_p\cup C_{\neg p})) &  \mid p\in \vars\}.
\end{align*}
	
	\begin{itemize}

	\item For each $X\in \clauses$, $1^X$ cannot join $\pi(2^X)$ because of the presence of $3^X$.  

	\item For each $X\in \clauses$, $2^X$ is only with $3^X$ and therefore
	in a most preferred coalition.

	\item For each $X\in \clauses$, $3^X$ cannot join $\pi(4^X)$ because of the presence of $5^X$. 

	\item For each $X\in \clauses$, $4^X$ is only with $5^X$ and therefore in a most preferred coalition.

	\item For each $X\in \clauses$, $5^X$ cannot join $\pi(1^X)$ because of
	the presence of a player $p^X$, for some $p\in \vars$. 

	\item For each $p\in \vars$, $0_p$ is together only 
	with a non-empty subset of $C_p\cup C_{\neg p}$, strictly preferring that
	to being alone.
	Any $p^X$ that  is not in the same coalition with $0_p$ is together 
	with $1^X$, therefore $0_p$ does not want to join.

	\item For each $p\in \vars$, either all players in $C_p$, or all
	players in $C_{\neg p}$ are together with $0_p$ and with no other players:
	assume that $C_{\neg p}$
	is the set of players with $0_p$ (the case $0_p\in C_p$ is proven similarly). 
	Then each player  $p^X\in C_p$ is
	together with $1^X$ and the players in $L^X_+$. They would prefer to
	join the coalition of $0_p$, but they are blocked by players in 
	$C_{\neg p}$. Players in $C_{\neg p}$ also cannot move to a better
	coalition, as they prefer their current coalition to being alone. 
	They would prefer a coalition only with players in $C_p\setminus\{p^X\}$,
	but all such players are either together with
	$0_p$, or with $1^Y$, for some $Y\neq X$.  
		
	\end{itemize}
	Therefore, $\pi$ is IS.

	For the other direction, assume that there is an IS
partition $\pi$. 
 Note first that for any $p\in \vars$ there can be no
$X,Y\in\clauses$, such that $p^X$ and $\neg p^Y$ are in the same coalition in 
$\pi$. Moreover,  not  both $p^X$ and $\neg p^Y$ are each together
with $1^X$ and $1^Y$, respectively. Assume that this was the case. Then, at least
 one of them
would be able to move to $0_p$ to take part in a more preferred coalition:
Either $0_p$ is alone and welcomes either of them, or it is together
with players either from $C_p$, or from $C_{\neg p}$, 
which would welcome $p$ or $\neg p$, respectively.
Therefore, for any variable $p$, either only members of $C_p$ or only members 
of $C_{\neg p}$ can be together with players $1^{X_1}, \ldots, 1^{X_k}$.
Notice now that there can be no clause $X^*\in \clauses$, such that
no $p^{X^*}$ is together with $1^{X^*}$, for any $p\in \vars$. 
For the sake of contradiction, assume that this were the case.
Then $1^{X^*}$ would be alone, and $5^{X^*}$ would
break off his coalition to join $1^{X^*}$. But then this would lead to
a series of deviations, with $1^{X^*}$ breaking the coalition with $5^{X^*}$
to join $2^{X^*}$, etc. Therefore, $\pi$ would not be an IS partition.
Therefore, for each $X\in \clauses$, there is some $p^X$ that is together 
with $1^X$. 
Consider now the following valuation $v_\pi$: 
For any $p\in \vars$, assign $v_\pi(p)=\textrm{true}$, if and only if 
 there is at least one player $p^X$ in $\pi$ that is in the same coalition 
as $1^X$, for any $X\in \clauses$. From the above, $v_\pi$ is a  
 a valid truthful assignment for $\form$, i.e., $\form$ is satisfiable. 
\end{proof}


The proof of the previous statement can also be used to state the following.

\begin{theorem}\label{theorem:ww-bb-srict-ns-hard}
Checking whether an NS partition exists is NP-complete for \BB-hedonic games or \WW-hedonic games with
strict preferences.
\end{theorem}
 
Finally, we have the following result.

\begin{theorem}\label{theorem:w-is-hard}
Checking whether an IS partition exists is NP-complete for \W-hedonic games.
\end{theorem}
\begin{proof}
The reduction is same as in the proof of Theorem~\ref{theorem:b-is-hard} except that preferences of certain players are as follows:

\begin{align*}
p^X 	&: (0_p, C_p\setminus\{p^X\} \mid 1^X, P\setminus(C_p\cup C_{\neg p})
	  \mid p^X  \middd \allelse) 	\\
\neg p^X &: (0_p, C_{\neg p}\setminus\{\neg p^X\} \mid 1^X, 
		P\setminus(C_p\cup C_{\neg p}) \mid \neg p^X   \middd \allelse) 	\\
0_p &: (C_p, C_{\neg p}, 0_p  \middd \allelse) 
\end{align*}
\end{proof}

It is not known whether Theorem~\ref{theorem:w-is-hard} also holds for \W-hedonic games with strict preferences.
However, there is a certain condition under which an IS partition can be computed in polynomial time. The roommate problem is a generalization of the stable marriage problem in which each agent has preferences over the other agents and then the agents are paired up in a stable manner~\citep{Irvi85a}. The problem can be seen as hedonic game called \emph{roommate game} in which only coalitions of size one or two are acceptable. We see that the same preferences of players over other players can represent roommate games and also \W-hedonic games. 
For roommate games with strict preferences, a core stable matching corresponds to a strict core stable partition for the corresponding \W-hedonic game~\citep{CeRo01a}. 
 Therefore, the algorithm of \citet{Irvi85a} can be used to check whether a core stable matching exists for the roommate games with strict preferences and if it exists, then the matching is a strict core stable and thereby IS partition for the \W-hedonic game with the same preference profile.

\section{\B-hedonic games: Nash \& individual stability}

We start with an easy observation.

\begin{observation}\label{obs:B-grand}
For \B-hedonic games, if each player likes at least some player, then the partition consisting of the grand coalition is Nash stable.
\end{observation}

For \B-hedonic games, even if preferences allow no unacceptability of players, a NS partition may still not exist. Consider a two-player \B-hedonic game in which player 1 finds player 2 simply acceptable and 2 likes 1. Then the grand coalition is not NS. Although, NS partitions may not exist for \B-hedonic games, one can efficiently check whether a NS partition exists if preferences are strict. The result holds not only for strict preferences but also for preferences in which there may be ties but they satisfy the \emph{unique favorite property} (if a player likes some other players, he has a unique favorite player).

\begin{theorem}\label{th-ns-strict-b}
For \B-hedonic games which satisfy the unique favorite property, there exists a linear time algorithm to check whether a NS partition exists or not. 
\end{theorem}
\begin{proof}
If each player likes at least one (acceptable) player, then the partition consisting of the grand coalition is NS and we are done. If there exists a set of players $A$ who find everyone else at best acceptable, form the partition $\{N\setminus A\} \cup \{\{i\}\mid i\in A\}$. Note that each player $j\in N\setminus A$ likes at least one other player $k\in N$. Also any NS partition will have players in $A$ as singletons. Now for any $j\in N\setminus A$, if $j$'s favorite player is in $A$, then return `no'. Else, return the partition $\pi= (N\setminus A) \cup \{\{i\}\mid i \in A\}$.

We first show that if there exists a player $j\in N\setminus A$ who has a favorite player in $A$, then there does not exist an NS partition. In any NS partition, players in $A$ would be singletons but if a player $j\in N\setminus A$ has a favorite player in $A$, then $j$ proves to be a stalker of some singleton player $a\in A$.   

We now show that if a NS partition exists, then partition $\pi$ is NS. Each player in set $(N\setminus A)$ has his (unique) favorite player in coalition $N\setminus A$. Also, the singleton players of $A$ do not like any other player so would rather remain alone.
\end{proof}

For the general preferences, as long as every player likes at least one other player, the grand coalition is NS. 
The problem becomes challenging if the unique favorite property is not satisfied and there exist players which do not like any other player. Although the case with no preference restrictions is open, it is equivalent to a restricted problem. 

\begin{theorem}\label{th-ns-strict-b}
The complexity of checking the existence of a NS partition for \B-hedonic games is equivalent to the same problem for \B-hedonic games with no unacceptability. 
\end{theorem}

\begin{proof}
	The general problem is at least as hard as with the restriction of no unacceptability. 
If there exists a polynomial-time algorithm to solve the case for no unacceptability, then the general problem can also be solved in polynomial time: change the preference profile in the general problem so that unacceptable players are now acceptable but not liked. Then a partition is NS in the restricted case if and only if it is NS in the general case. 
\end{proof}

Whereas NS partitions may not exist for \B-hedonic games, we give a constructive argument for the existence of IS partitions. Therefore, we add \B-hedonic games to the following list of hedonic games and preference restrictions for which IS partition is guaranteed to exist: additively separable hedonic games with symmetric preferences and anonymous games with single-peaked preferences~\citep{BoJa02a}. This is contrast to B-hedonic games for which even checking the existence of an IS partition is NP-complete.

\begin{theorem}\label{th:is-exists}
For \B-hedonic games, an IS partition exists and it can be computed in linear time. 
\end{theorem}

\begin{proof}
If each player likes at least one other player, then the partition consisting of the grand coalition is IS and we are done. 
Otherwise, we maintain a variable set of players $B$ which will eventually converge to a fixed set and help return an IS partition. We also maintain a variable partition $\pi'=\{N\setminus B\} \cup \{\{i\}\mid i\in B\}$. 

If there exists a set of players $A$ who do not like any other player, then set the variable set of players $B$ to $A$ and form the partition $\pi'=\{N\setminus B\} \cup \{\{i\}\mid i\in B\}$. If there exists a player $j\in N\setminus B$ such that $j$ now also does not like any players in $N\setminus B$, then set $B$ to $A\cup \{j\}$ and update the partition $\pi'$ to $\{N\setminus B\} \cup \{\{i\}\mid i\in B\}$ where the updated $B$ is used. Since $j$ was not liked by any player in $A$, he is not allowed to join any of them. Repeat this process until there exists no $j'\in N\setminus B$ who likes some player in $N\setminus B$. Once, $B$ cannot be updated anymore, we return $\pi$ the final value of $\pi'=\{N\setminus B\} \cup \{\{i\}\mid i\in B\}$. 

We now show that the returned partition $\pi$ is IS. We do so by showing that none of the players in the following sets can have a feasible IS deviation: 1. $N\setminus B$, 2. $A$, and 3. $B\setminus A$. 

\begin{enumerate}
	\item Players in $N\setminus B$ like at least one other player in $N\setminus B$ and therefore would not prefer to become alone. Furthermore, players in $N\setminus B$ cannot join a player in $B$: no singleton player in $B$ is willing to welcome a player in $N\setminus B$. 
\item Players in $A\subseteq B$ do not like any one and prefer to stay alone. \item Each player in $x\in B\setminus A$ does not like any other player $y\in B\setminus A$ who left $N\setminus B$ later than $x$. Therefore $x$ will not welcome $y$ even if $y$ wanted to join $x$ and form coalition $\{x,y\}$. Therefore there is no deviation by singleton players in $B\setminus A$.
\end{enumerate}

We have shown that no player in $N$ has a feasible IS deviation. Thus $\pi$ is IS.
This completes the proof.
\end{proof}

We saw in this section that \B-hedonic games not only admit an IS partition but if the `unique favorite property' is satisfied, it can also be checked efficiently whether a NS partition exists.

\section{Conclusions}

In this paper, we analyzed hedonic coalition formation games based on the best or worst players. In particular, we considered the existence and computation of individual-based stability concepts in \B-, \W, \BB- and \WW-hedonic games. An almost complete characterization of the complexity of computing stable partitions was achieved (see Table~\ref{table:bestworst-complexity}). 
We showed that for \B-hedonic games, there exists a polynomial-time algorithm which returns an IS partition. For strict preferences, it can be checked in polynomial time whether a NS partition exists. For all other games, checking the existence of NS or IS partitions is intractable even for many cases when preference are strict. 
An open problem is the complexity of checking the existence of an IS partition for \W-hedonic games when preferences are strict. 
It was seen that in coalition formation games based on extensions of preferences over players to preferences over sets, the following property is a major source of intractability:
presence of an unacceptable player makes the coalition unacceptable.

Future directions of research include finding further restrictions on the preferences so that stability is guaranteed or is computationally feasible to analyze. For example, one can consider the case where there is a fixed ordering of the players (based on publicly known trait of the players) and each player has single peaked preferences in the ordering. 
Another interesting setting is when there exists a fixed global ordering of players and each player $i$ prefers more those players that are closer to $i$ in the ordering. Furthermore one may consider scenarios with communication restrictions among players, represented via a graph, such that only contiguous players can form a coalition. 




\begin{thebibliography}{22}
\providecommand{\natexlab}[1]{#1}
\providecommand{\url}[1]{\texttt{#1}}
\expandafter\ifx\csname urlstyle\endcsname\relax
  \providecommand{\doi}[1]{doi: #1}\else
  \providecommand{\doi}{doi: \begingroup \urlstyle{rm}\Url}\fi

\bibitem[Aziz et~al.(2011{\natexlab{a}})Aziz, Brandt, and Harrenstein]{ABH11b}
H.~Aziz, F.~Brandt, and P.~Harrenstein.
\newblock Pareto optimality in coalition formation.
\newblock In G.~Persiano, editor, \emph{Proceedings of the 4th International
  Symposium on Algorithmic Game Theory (SAGT)}, Lecture Notes in Computer
  Science (LNCS), pages 93--104. Springer-Verlag, 2011{\natexlab{a}}.

\bibitem[Aziz et~al.(2011{\natexlab{b}})Aziz, Brandt, and Seedig]{ABS11b}
H.~Aziz, F.~Brandt, and H.~G. Seedig.
\newblock Stable partitions in additively separable hedonic games.
\newblock In P.~Yolum and K.~Tumer, editors, \emph{Proceedings of the 10th
  International Joint Conference on Autonomous Agents and Multi-Agent Systems
  (AAMAS)}, pages 183--190. IFAAMAS, 2011{\natexlab{b}}.

\bibitem[Ballester(2004)]{Ball04a}
C.~Ballester.
\newblock {NP}-completeness in hedonic games.
\newblock \emph{Games and Economic Behavior}, 49\penalty0 (1):\penalty0 1--30,
  2004.

\bibitem[Banerjee et~al.(2001)Banerjee, Konishi, and S{\"o}nmez]{BKS01a}
S.~Banerjee, H.~Konishi, and T.~S{\"o}nmez.
\newblock Core in a simple coalition formation game.
\newblock \emph{Social Choice and Welfare}, 18:\penalty0 135--153, 2001.

\bibitem[Barber{\`a} et~al.(2004)Barber{\`a}, Bossert, and Pattanaik]{BBP04a}
S.~Barber{\`a}, W.~Bossert, and P.~K. Pattanaik.
\newblock Ranking sets of objects.
\newblock In S.~Barber{\`a}, P.~J. Hammond, and C.~Seidl, editors,
  \emph{Handbook of Utility Theory}, volume~II, chapter~17, pages 893--977.
  Kluwer Academic Publishers, 2004.

\bibitem[Bogomolnaia and Jackson(2002)]{BoJa02a}
A.~Bogomolnaia and M.~O. Jackson.
\newblock The stability of hedonic coalition structures.
\newblock \emph{Games and Economic Behavior}, 38\penalty0 (2):\penalty0
  201--230, 2002.

\bibitem[Cechl{\'a}rov{\'a}(2008)]{Cech08a}
K.~Cechl{\'a}rov{\'a}.
\newblock Stable partition problem.
\newblock In \emph{Encyclopedia of Algorithms}, pages 885--888. Springer, 2008.

\bibitem[Cechl{\'a}rov{\'a} and Hajdukov{\'a}(2002)]{CeHa02a}
K.~Cechl{\'a}rov{\'a} and J.~Hajdukov{\'a}.
\newblock {Computational complexity of stable partitions with B-preferences}.
\newblock \emph{International Journal of Game Theory}, 31\penalty0
  (3):\penalty0 353--354, 2002.

\bibitem[Cechl{\'a}rov{\'a} and Hajdukov{\'a}(2004)]{CeHa04a}
K.~Cechl{\'a}rov{\'a} and J.~Hajdukov{\'a}.
\newblock Stable partitions with {$\mathcal{W}$}-preferences.
\newblock \emph{Discrete Applied Mathematics}, 138\penalty0 (3):\penalty0
  333--347, 2004.

\bibitem[Cechl{\'a}rov{\'a} and Romero-Medina(2001)]{CeRo01a}
K.~Cechl{\'a}rov{\'a} and A.~Romero-Medina.
\newblock Stability in coalition formation games.
\newblock \emph{International Journal of Game Theory}, 29:\penalty0 487--494,
  2001.

\bibitem[Dr{\`e}ze and Greenberg(1980)]{DrGr80a}
J.~H. Dr{\`e}ze and J.~Greenberg.
\newblock Hedonic coalitions: Optimality and stability.
\newblock \emph{Econometrica}, 48\penalty0 (4):\penalty0 987--1003, 1980.

\bibitem[Elkind and Wooldridge(2009)]{ElWo09a}
E.~Elkind and M.~Wooldridge.
\newblock Hedonic coalition nets.
\newblock In \emph{Proceedings of the 8th International Joint Conference on
  Autonomous Agents and Multi-Agent Systems (AAMAS)}, pages 417--424, 2009.

\bibitem[Gairing and Savani(2010)]{Gasa10a}
M.~Gairing and R.~Savani.
\newblock Computing stable outcomes in hedonic games.
\newblock In S.~Kontogiannis, E.~Koutsoupias, and P.~Spirakis, editors,
  \emph{Proceedings of the 3rd International Symposium on Algorithmic Game
  Theory (SAGT)}, volume 6386 of \emph{Lecture Notes in Computer Science},
  pages 174--185. Springer-Verlag, 2010.

\bibitem[Gairing and Savani(2011)]{Gasa11a}
M.~Gairing and R.~Savani.
\newblock Computing stable outcomes in hedonic games with voting-based
  deviations.
\newblock In \emph{Proceedings of the 10th International Joint Conference on
  Autonomous Agents and Multi-Agent Systems (AAMAS)}, pages 559--566, 2011.

\bibitem[Garey and Johnson(1979)]{GaJo79a}
M.~R. Garey and D.~S. Johnson.
\newblock \emph{Computers and Intractability: A Guide to the Theory of
  NP-Completeness}.
\newblock W. H. Freeman, 1979.

\bibitem[Hajdukov{\'a}(2006)]{Hajd06a}
J.~Hajdukov{\'a}.
\newblock Coalition formation games: {A} survey.
\newblock \emph{International Game Theory Review}, 8\penalty0 (4):\penalty0
  613--641, 2006.

\bibitem[Irving(1985)]{Irvi85a}
R.~W. Irving.
\newblock An efficient algorithm for the ``stable roommates'' problem.
\newblock \emph{Journal of Algorithms}, 6\penalty0 (4):\penalty0 577--595,
  1985.

\bibitem[Olsen(2009)]{Olse09a}
M.~Olsen.
\newblock Nash stability in additively separable hedonic games and community
  structures.
\newblock \emph{Theory of Computing Systems}, 45:\penalty0 917--925, 2009.

\bibitem[Roth and Sotomayor(1990)]{RoSo90a}
A.~Roth and M.~A.~O. Sotomayor.
\newblock \emph{Two-Sided Matching: {A} Study in Game Theoretic Modelling and
  Analysis}.
\newblock Cambridge University Press, 1990.

\bibitem[Sandholm et~al.(1999)Sandholm, Larson, Andersson, Shehory, and
  Tohm{\'e}]{SLA+99a}
T.~Sandholm, K.~Larson, M.~Andersson, O.~Shehory, and F.~Tohm{\'e}.
\newblock Coalition structure generation with worst case guarantees.
\newblock \emph{Artificial Intelligence}, 111\penalty0 (1--2):\penalty0
  209--238, 1999.

\bibitem[Shapley and Scarf(1974)]{ShSc74a}
L.~S. Shapley and H.~Scarf.
\newblock On cores and indivisibility.
\newblock \emph{Journal of Mathematical Economics}, 1\penalty0 (1):\penalty0
  23--37, 1974.

\bibitem[Sung and Dimitrov(2010)]{SuDi10a}
S.~C. Sung and D.~Dimitrov.
\newblock Computational complexity in additive hedonic games.
\newblock \emph{European Journal of Operational Research}, 203\penalty0
  (3):\penalty0 635--639, 2010.

\end{thebibliography}

\end{document}